\documentclass[10pt,conference]{IEEEtran}
\normalsize
\usepackage[T1]{fontenc}
\usepackage{amsmath,amssymb,amsfonts}
\usepackage{mathrsfs}
\usepackage{mathabx}
\usepackage{amsbsy}
\usepackage{graphicx}
\usepackage{graphics}
\usepackage{epstopdf}
\usepackage{algorithm}
\usepackage{algorithmic}
\usepackage{subfigure}
\usepackage{cite}
\usepackage{slashbox}
\usepackage{multirow}
\usepackage{ctable}
\usepackage{tikz,pgfplots}

\newtheorem{theorem}{{Theorem}}
\newtheorem{lemma}[theorem]{{Lemma}}

\newtheorem{corollary}[theorem]{{Corollary}}

\newtheorem{definition}{{Definition}}

\newcommand{\cA}{{\cal A}} 
\newcommand{\cB}{{\cal B}}
\newcommand{\cC}{{\cal C}}

\newcommand{\cG}{{\cal G}} 
\newcommand{\cH}{{\cal H}}

\DeclareMathAlphabet{\mathbfsl}{OT1}{ppl}{b}{it} %{OT1}{cmr}{bx}{it}

\newcommand{\bU}{\mathbfsl{U}}

\newcommand{\bY}{\mathbfsl{Y}}

\def\QEDclosed{\mbox{\rule[0pt]{1.3ex}{1.3ex}}} % for a filled box
\def\QED{\QEDclosed} % default to closed

\newcommand{\be}[1]{\begin{equation}\label{#1}}
\newcommand{\ee}{\end{equation}} 
\newcommand{\eq}[1]{(\ref{#1})}

\renewcommand{\le}{\leqslant} 
\renewcommand{\leq}{\leqslant}
\renewcommand{\ge}{\geqslant} 
\renewcommand{\geq}{\geqslant}

\newcommand{\script}[1]{{\mathscr #1}}

\renewcommand{\Bbb}{\mathbb}
\newcommand{\C}{{\Bbb C}}

\newcommand{\Tref}[1]{Theo\-rem\,\ref{#1}}

\newcommand{\Lref}[1]{Lem\-ma\,\ref{#1}}
\newcommand{\Cref}[1]{Co\-ro\-lla\-ry\,\ref{#1}}

\newcommand{\deff}{\mbox{$\stackrel{\rm def}{=}$}}
\newcommand{\Strut}[2]{\rule[-#2]{0cm}{#1}}

\newcommand{\hU}{\widehat{\bU}}

\newcommand{\Gn}{G^{\otimes n}}

\newcommand{\Span}[1]{{\left\langle {#1} \right\rangle}}

\newcommand{\sX}{\script{X}}
\newcommand{\sY}{\script{Y}}

\newcommand{\shalf}{\mbox{\raisebox{.8mm}{\footnotesize $\scriptstyle 1$}
\footnotesize$\!\!\! / \!\!\!$ \raisebox{-.8mm}{\footnotesize
$\scriptstyle 2$}}}

\begin{document}

\title{Compound Polar Codes}

\author{%
\authorblockN{\large{
Hessam Mahdavifar, %~\IEEEmembership{Member,~IEEE,}
        Mostafa El-Khamy, %~\IEEEmembership{Member,~IEEE,}
        Jungwon Lee,
        Inyup Kang}}\\
\authorblockA{
%Department of Electrical and Computer Engineering,
Mobile Solutions Lab, Samsung Information Systems America\\
4921 Directors Place, San Diego, CA 92121\\
{ \{h.mahdavifar,\,mostafa.e,\,jungwon2.lee,\,inyup.kang\}@samsung.com}}\vspace*{-1.5ex}
}

\renewcommand{\markboth}[2]
{\renewcommand{\leftmark}{#1}\renewcommand{\rightmark}{#2}}

\markboth%
{{\sc Mahdavifar, El-Khamy, Lee and Kang:}
COMPOUND POLAR CODES}
{To be submitted to {\sc IEEE Transactions on Information Theory}}

\maketitle

\begin{abstract}
A capacity-achieving scheme based on polar codes is proposed for reliable communication over multi-channels which can be directly applied to bit-interleaved coded modulation schemes. We start by reviewing the ground-breaking work of polar codes and then discuss our proposed scheme. Instead of encoding separately across the individual underlying channels, which requires multiple encoders and decoders, we take advantage of the recursive structure of polar codes to construct a unified scheme with a single encoder and decoder that can be used over the multi-channels. We prove that the scheme achieves the capacity over this multi-channel. Numerical analysis and simulation results for BICM channels at finite block lengths shows a considerable improvement in the probability of error comparing to a conventional separated scheme. 
\end{abstract}


\IEEEpeerreviewmaketitle

%=======================================================================%
%                                                                       %
%    1. INTRODUCTION                                                    %
%                                                                       %
%=======================================================================%
\section{Introduction} 
\label{sec:Introduction}
\PARstart{P}{olar} codes, introduced by Arikan in \cite{Arikan}, are the first provably capacity achieving codes for the class of binary-input symmetric discrete memoryless channels with low encoding and decoding complexity. Construction of polar codes is based on a phenomenon called the \emph{channel polarization}. It is proved in \cite{Arikan} that as the block length grows large the channels seen by individual bits through a certain transformation called the \emph{polar transformation} start polarizing: they approach either a noise-less channel or a pure-noise channel. This suggests the construction of polar codes as follows: put the information bits over the set of good bit-channels i.e. almost noise-less channels while fix the input to the rest of the bit-channels to zeros. The set of underlying bit-channels can be sorted from good to bad based on their corresponding \emph{Bhattacharyya parameter}. Then the bit-channels with Bhattacharrya parameters below a certain threshold are called good while the rest are called bad. In fact, the unequal error protection seems to be an inherent property of polar codes. Therefore, it is natural to exploit this property to design codes for bit-interleaved coded modulation channels. 

Bit-interleaved coded modulation (BICM) schemes can be modeled as a multi-channel consisting of several underlying binary-input channels over which the coded bits are transmitted. Therefore, it is of great interest to design codes that are efficient when used over a certain set of different channels. However, it is not straightforward how to design a polar code to be transmitted over a set of channels rather than a single channel, how to establish polarization theory in this case etc. One immediate solution is to encode the information separately over the underlying channels using polar encoders corresponding to each channel. However, from a practical point of view, it is desirable to have only one encoder and one decoder to reduce the hardware complexity. Also, by combining all the channels together and sending one single codeword, efficiently designed for the corresponding multi-channel, we achieve a better trade-off between the rate and probability of error in the whole scheme. We call this unified polar-based scheme for transmission over multi-channels as \emph{compound polar codes}.

The rest of this paper is organized as follows. In Section\,\ref{sec:two}, we provide some required background on channel polarization and construction of polar codes. In Section\,\ref{sec:three}, we propose compound polar codes for $2$-multi-channels and prove that it achieves the capacity of the combined channel. We also discuss how to extend this to $l$-multi-channels. In Section\,\ref{sec:four}, we present the simulation results for AWGN channel with BICM $16$-QAM constellation. We close the paper by mentioning some directions for future work and open problems in Section\,\ref{sec:five}.

%=======================================================================%
%                                                                       %
%     2. Preliminaries                                                  %              
%                                                                       %
%=======================================================================%
\section{Preliminaries}
\label{sec:two}

In this section we provide an overview of the groundbreaking work of
Arikan~\cite{Arikan} 
and others~\cite{AT,Korada,KSU} on polar 
codes and channel polarization. 

The polar code construction is based on the following observation by Arikan which is called \emph{channel polarization}. Let 
\be{G-def}
G 
\ = \
\left[ 
\begin{array}{c@{\hspace{1.25ex}}c}
1 & 0\\
1 & 1\\ 
\end{array}
\right]
\ee
The $i$-th Kronecker power of $G$, which is denoted by $G^{\otimes i}$, is defined by induction i.e. $G^{\otimes 1} = G$ and for any $i > 1$:
$$
G^{\otimes (i)}
\ = \
\left[ 
\begin{array}{c@{\hspace{1.25ex}}c}
G^{\otimes (i-1)} & 0\\
G^{\otimes (i-1)} & G^{\otimes (i-1)}\\ 
\end{array}
\right]
$$
Next, for all $N = 2^n$, let us define the Arikan transform matrix
$G_N \,\smash{\raisebox{-0.35ex}{\deff}}\, R_N \Gn$, 
where $R_N$ is the bit-reversal permutation matrix defined 
in \cite[Section\,VII-B]{Arikan}. Now consider a block of $N$ uniform i.i.d. information bits $U_1, U_2,\dots,U_N$, denoted by $U_1^N$, and multiply it by $G_N$ to get the vector $X_1^N$. $X_i$'s are transmitted through $N$ independent copies of a binary input discrete memoryless channel (B-DMC) $W$. The output is denoted by $Y_1^N$. The transformation from $U^N_1$ to $Y^N_1$ is called the polar transformation. 

A finite-input and finite-output discrete memoryless channel is denoted by a triple $\Span{\sX,\sY,W}$, where $\sX$, $\sY$
are finite sets. $W$ is the transition probability matrix which is an $|\sX| \times |\sY|$ matrix. For any $x \in \!\sX$ and $y \in \!\sY$, $W[x,y]$, conventionally written as $W(y|x)$, is the probability of receiving $y \in \!\sY$ given that $x \in \!\sX$ was sent. With a slight abuse of notation, we simply write $W$ to 
denote the channel $\Span{\sX,\sY,W}$.

\begin{definition} 
A binary-input discrete memoryless channel (B-DMC) $W : \left\{0,1\right\} \rightarrow \sY$ is called \emph{symmetric} if there exists a permutation $\pi : \sY \rightarrow \sY$ with $\pi^{-1} = \pi$ such that for any $y \in \sY$, $W(y | 0) = W(\pi(y) | 1)$.
\end{definition}

For any B-DMC $W$, the \emph{Bhattacharyya parameter} of $W$ is
$$ 
Z(W)
\,\ \deff\kern1pt
\sum_{y\in\sY} \!\sqrt{W(y|0)W(y|1)}
$$
It is easy to show that the Bhattacharyya parameter $Z(W)$ is always between $0$ and $1$.
Intuitively, $Z(W)$ shows how good the channel $W$ is. Channels with $Z(W)$ close to zero are 
almost noiseless, while channels with $Z(W)$ close to one 
are almost pure-noise channels.
This intuition is clarified more by the following inequality. It is shown in \cite{Arikan} that for any B-DMC $W$,
\be{Z-I}
1 - I(W) \leq Z(W) \leq \sqrt{1 - I(W)^2}
\ee
where $I(W)$ is the symmetric capacity of $W$. 

Let $W^N$ denote the channels that results from $N$ independent copies of $W$ i.e. the channel
$\bigl\langle\{0,1\}^N,\sY^N\hspace{-1pt},W^N\bigr\rangle$
given by\vspace{-1.00ex}
\be{Wn}
W^N\kern-0.5pt(y^N_1|x^N_1) 
\,\ \deff\,\
\prod_{i=1}^N W(y_i|x_i)
\vspace{-0.25ex}
\ee
where $x^N_1 \hspace{1pt}\,{=}\, (x_1,x_2,\dots,x_N)$ and
$y^N_1 \hspace{1pt}\,{=}\, (y_1,y_2,\dots,y_N)$. Then the \emph{combined} channel
\smash{$\bigl\langle\{0,1\}^N,\sY^N\hspace{-1pt},\widetilde{W}\bigr\rangle$} is defined
with transition probabilities
given by
\be{Wtilde}
\widetilde{W}(y^N_1|u^N_1) 
\,\ \deff\,\
W^N\kern-1pt\bigl(y^N_1\hspace{1pt}{\bigm|}\hspace{1pt}u^N_1\hspace{1pt} G_N\bigr)
\kern1pt = \kern2pt
W^N\kern-1pt
\bigl(y^N_1\hspace{1pt}{\bigm|}\hspace{1pt}u^N_1\hspace{1pt} R_N \Gn \bigr)
\ee
This is the channel that the random vector
$(U_1,U_2,\dots,U_N)$ observes through the polar transformation defined earlier.
Arikan~\cite{Arikan} also defines 
the $i$-th bit-channel
$\smash{\bigl\langle\{0,1\},\sY^N{\times}\{0,1\}^{i-1},W^{(i)}_N\bigr\rangle}$, for $i = 1,2,\dots,N$,
as follows. Let $u^i_1 = (u_1,u_2,\dots,u_i)$ denote a binary vector
of length $i$. For $i = 0$, this is the empty string. 
Then
\be{Wi-def}
W^{(i)}_N\bigl( y^N_1,u^{i-1}_1 | \hspace{1pt}u_i)
\deff
\frac{1}{2^{N-1}}\hspace{-10pt}
\sum_{u_{i+1}^N \in \{0,1\}^{N-i}} \hspace{-12pt}
\widetilde{W}\Bigl(y^N_1\hspace{1pt}{\bigm|}\hspace{1pt}
(u^{i-1}_1,u_i,u_{i+1}^N) \Bigr)
\ee
It can be shown %~(cf.\ \Lref{lemma8})
that $W^{(i)}_N\bigl( y^N_1,u^{i-1}_1 | \hspace{1pt}u_i)$
is indeed the probability of the event that
$(Y_1,Y_2,\dots,Y_N) \hspace{-1pt}= y^N_1\/$ and 
$(U_1,U_2,\dots,U_{i-1}) \hspace{-1pt}= u^{i-1}_1$
given the event $U_i = u_i$, provided
$U^N_1$ is a priori uniform over $\{0,1\}^N$. Intuitively, this is the channel that bit $u_i$ observes under Arikan's \emph{successive cancellation decoding}, described later. 

The $N$ bit-channels are partitioned
into \emph{good channels} and \emph{bad channels} as follows ~\cite{AT,Korada}. Let 
\smash{$[N] \ \raisebox{-0.2ex}{\deff}\ \{1,2,\dots,N\}$}\Strut{2.15ex}{0ex}
and let $\beta \,{<}\, \shalf$ be a fixed positive constant.
Then the index sets of the good and bad channels are given by
\begin{eqnarray}
\label{good-def}
\cG_N(W,\beta)
&\hspace*{-6pt}{\deff}\hspace*{-6pt}&
\left\{\, i \in [N] ~:~ Z(W^{(i)}_N) < 2^{-N^{\beta}}\!\!/N \hspace{1pt}\right\}
\\[0.250ex]
\label{bad-def}
\cB_N(W,\beta)
&\hspace*{-6pt}{\deff}\hspace*{-6pt}&
\left\{\, i \in [N] ~:~ Z(W^{(i)}_N) \ge 2^{-N^{\beta}}\!\!/N \hspace{1pt}\right\}
\\[-2.75ex]
\nonumber
\end{eqnarray}

\begin{theorem}
\label{thm1}
~\cite{Arikan,AT} For any binary symmetric memoryless (BSM) channel $W$ and any constant $\beta \,{<}\, \shalf$ we have
$$
\lim_{N \to \infty} \frac{\left|\cG_N(W,\beta)\right|}{N}
\,=\, 
\cC(W)\vspace{1.5ex}
$$
\end{theorem}

\Tref{thm1} readily leads to a construction of capacity-achieving \emph{polar codes}. The idea is to transmit the information bits over the good bit-channels while fixing the input to the bad 
bit-channels to a priori known values, say zeros. 
Formally, each subset $\cA$ of $[N]$ of size $|\cA| = k$ specifies
a \emph{polar code} $\C_N(\cA)$ of rate $k/N$.
$\C_N(\cA)$ is actually a linear code with length $N$ and dimension $k$. 
The generator matrix of $\C_N(\cA)$ is a $k \times N$ matrix that consists of rows of $G_N$ corresponding to
the elements of $\cA$. 

Arikan also introduces the successive cancellation decoding for polar codes which leads to the following key theorem on the encoder-decoder pair of polar codes. This theorem is (the second part of) 
Proposition\,2 of Arikan~\cite{Arikan}.
\begin{theorem}
\label{thm2}
Let $W$ be a BSM channel and let $\cA$ be an
arbitrary subset of $[N]$ of size $|\cA| = k$.
Suppose that a message $\bU$ is chosen 
uniformly at random from $\{0,1\}^k$,
encoded as a codeword of\/ $\C_N(\cA)$, and transmitted
over $W$. Then the probability that the channel
output is {not} decoded to $\bU$ under successive
cancellation decoding satisfies
\be{eq3}
\Pr \bigl\{\hU \ne \bU \bigr\} 
\, \le \,
\sum_{i \in \cA} \!Z(W^{(i)}_N) 
%R \ \le \ 2^{-n^{\beta}} 
%\vspace{-.50ex}
\ee
\end{theorem}

\begin{corollary}
\label{cor1}
For any $\beta \,{<}\, \shalf$ and any BSM channel $W$, the polar code of length $N$ associated with the set of good bit-channels $\cG_N(W,\beta)$ defined in \eq{good-def} approaches the capacity of $W$. Furthermore, the probability of frame error under successive cancellation decoding is less than $2^{-N^{\beta}}$. 
\end{corollary}

%=======================================================================%
%                                                                       %
%       3. Compound polar codes for multi-channels               %              
%                                                                       %
%=======================================================================%
\section{Compound Polar Codes for multi-channels}
\label{sec:three}

In this section, we start with explaining the model for multi-channels. Then we describe our compound polar construction for the case of $2$-multi-channels. We extend compound polar code to the case of $l$ parallel channels and prove the capacity-achieving property. Also, the successive cancellation decoding and its complexity for compound polar codes is discussed. 

\subsection{Multi-channels}

We consider the following model for a multi-channel consisting of several binary-input discrete memoryless channels. Let  $W_i : \sX \rightarrow \sY_i$, for $i=1,2,\dots,l$, denote the $l$ given B-DMCs (indeed $\sX = \left\{0,1\right\}$). Then the corresponding $l$-multi-channel $(W_1.W_2\dots W_l) : \sX^l \rightarrow \sY_1 \times \sY_2 \times \dots \times \sY_l$ is another DMC whose transition probability for $x_1^l = (x_1,x_2,\dots,x_l) \in \sX^l$ and $y_1^l = (y_1,y_2,\dots,y_l) \in \sY_1 \times \sY_2 \times \dots \times \sY_l$ is given by:
$$
(W_1.W_2\dots W_l) \bigl(y_1^l | x_1^l\bigr) = \prod_{i=1}^l W_i(y_i|x_i)
$$ 
In fact, each binary sequence of length $l$ is transmitted through this multi-channel in such a way that each bit is transmitted over one of the $l$ channels. In general, for any $N$ which is a multiple of $l$, a sequence of $N$ bits is transmitted over this multi-channel in such a way that each channel carries $N/l$ bits. It is known to both transmitter and receiver that which channel carries which bits in the sequence. An arbitrary interleaver and deinterleaver can be employed in the transmitter and the receiver and thus, the ordering of channels does not matter. 

In this section, we first explain a straightforward scheme which encodes and decodes separately over the underlying channels. Then we discuss our proposed scheme of compound polar codes and prove the channel polarization theorem for this scheme.

\subsection{A straightforward construction}

One straightforward solution for constructing polar code over a multi-channel is to encode the information separately over the underlying binary-input channels using polar encoders corresponding to each channel. Suppose that a set of $l$ channels $W_1,W_2,\dots,W_l$ is given. We want to construct a scheme of length $N$ and rate $R$ for transmission over this set of channels such that all the channels are used equally $N/l$ times. Also, we require that $R$ approaches the average capacity of all the channels as $N$ goes to infinity. For $i = 1,2,\dots,l$, we construct a polar code of length $N/l$ with rate $R_i$ to be transmitted over $W_i$. Let 
$$
R = \frac{1}{l} (R_1 + R_2 + \dots + R_l)
$$
Then the scheme of length $N$ and rate $R$ is as follows. Given the input sequence of $NR$ bits, split it into $l$ chunks of size $NR_i/l$, for $i = 1, 2, \dots, l$. Then encode the $i$-th chunk using the $i$-th polar encoder and transmit the encoded sequence over $W_i$. There are also $l$ separated decoders at the receiver to decode the output of each of the channels separately to get the $i$-th transmitted chunk. We call this scheme the \emph{separated scheme}. Assuming that all the underlying channels are BSM channels, the rate $R_i$ of the constructed polar code over $W_i$ approaches the capacity of $W_i$ for all $i$'s. Therefore $R$ approaches the average of the capacities of $W_i$'s. 

\subsection{Compound polar transformation over $2$-multi-channels}

In this section, we propose our unified scheme for construction of polar codes over $2$-multi-channels. From a practical point of view, it is desirable to have only one encoder and one decoder to reduce the hardware complexity. Also, by combining all the channels together and sending one single codeword, efficiently designed for this multi-channel, we achieve a better trade-off between the rate and probability of error in the whole scheme. 

Suppose that a multi-channel with two constituent B-DMC channels $W_1$ and $W_2$ is given. The building block of our compound polar transformation is shown in Figure\,\ref{building_block}. 

\begin{figure}[h]
\begin{center}
\includegraphics[width=2in]{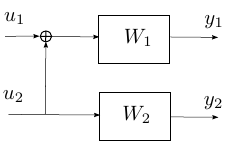}
\caption{The proposed building block of our scheme}
\label{building_block}
\end{center}
\end{figure}

We actually combine the two channels at the first step. Then the recursion is applied to this combined channel exactly same as polar codes. More precisely, let $\cH$ denote the channel with input $u_1$ and $u_2$ and output $y_1$ and $y_2$ as shown in Figure\,\ref{building_block}. Then there exist a permutation $\pi$, such that applying $G^{\otimes (n-1)}$ to $2^{n-1}$ independent copies of $\cH$ is equivalent to the transformation shown in Figure\,\ref{scheme}. The permutation block $\pi$ is designed in such a way that the first half of encoded block is transmitted through $W_1$ and the second half through $W_2$. The whole transformation from $u_1^N$ to $Y_1^N$ is called the compound polar transformation.

\begin{figure}[h]
\begin{center}
\includegraphics[width=\linewidth]{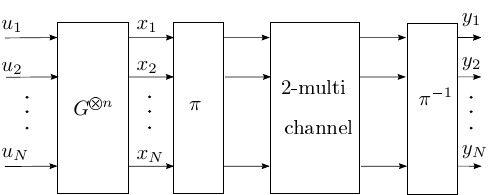}
\caption{The proposed scheme for $2$-multi-channels with length $N$}
\label{scheme}
\end{center}
\end{figure}

\subsection{Construction of capacity achieving compound polar codes}

The building block shown in Figure\,\ref{building_block} can split into two bit-channels by generalizing the definition of channel combining suggested in \cite{Arikan}. Suppose that $u_1$ and $u_2$ are samples of two independent uniform binary random variables $U_1$ and $U_2$, respectively. Notice that by applying the chain rule to $I(U_1^{2};Y_{1}^{2})$ we get
\begin{align*}
I(U_1^{2};Y_{1}^{2}) &= I(U_1;Y_1^{2}) + I(U_2;Y_1^2 | U_1) \\
& = I(U_1;Y_1^2) + I(U_2;Y_1^2,U_1)
\end{align*}
where the last equality follows since $U_1$ and $U_2$ are assumed to be independent. The term $I(U_1,Y_1^2)$ can be interpreted as the mutual information of the channel between $U_1$ and the output $Y_1^2$, where $U_2$ is considered as noise. Let us denote this channel by $W_1 \boxcoasterisk W_2$. Formally, for any two B-DMCs $W_1 : \sX \rightarrow \sY_1$ and $W_2 : \sX \rightarrow \sY_2$ (indeed $\sX = \left\{0,1\right\}$), let $W_1 \boxcoasterisk W_2 : \sX \rightarrow \sY_1 \times \sY_2$ denote another B-DMC whose transition probability for any $(y_1,y_2) \in \sY_1 \times \sY_2$ and $u \in \sX$ is given by
\be{operation1}
W_1 \boxcoasterisk W_2 (y_1,y_2 | u) = \frac{1}{2} \sum_{x \in \sX} W_1(y_1 | u \oplus x) W_2(y_2|u)
\ee
Similarly, the term $I(U_2;Y_1^2,U_1)$ can be interpreted as the mutual information of the channel between $U_2$ and $Y_1^2$ when $U_1$ is available at the decoder. Formally, for any two B-DMCs $W_1 : \sX \rightarrow \sY_1$ and $W_2 : \sX \rightarrow \sY_2$, let $W_1 \circledast W_2 : \sX \rightarrow \sY_1 \times \sY_2 \times \sX$ denote another B-DMC whose transition probability for any $(y_1,y_2) \in \sY_1 \times \sY_2$ and $ x, u \in \sX$ is given by
\be{operation2}
W_1 \circledast W_2 (y_1,y_2,x | u) = \frac{1}{2} W_1(y_1 | u \oplus x) W_2(y_2|u)
\ee 
The channels $W_1 \boxcoasterisk W_2$ and $W_1 \circledast W_2$ are depicted in Figure\,\ref{building_bit_channel}.

\begin{figure}[h]
\begin{center}
\includegraphics[width=\linewidth]{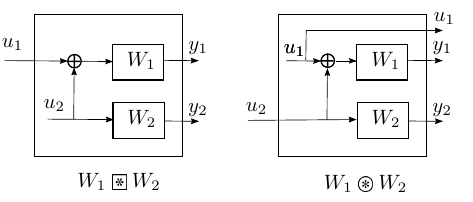}
\caption{The bit-channels of the proposed building block}
\label{building_bit_channel}
\end{center}
\end{figure}

The individual bit-channels can be defined for this compound polar transformation with some modification to the definition of bit-channels for original polar transformation in \eq{Wi-def}. Let $\widetilde{W}$ denote the channel from $u^N_1$ to $y^N_1$ in Figure\,\ref{scheme}. Then the individual bit-channels $(W_1.W_2)^{(i)}_N$ are defined as follows. For $i = 1, 2,\dots,N$,
\be{Wi-def-new}
\begin{split}
&(W_1.W_2)^{(i)}_N\bigl( y^N_1,u^{i-1}_1 | \hspace{1pt}u_i)\\
&\,\ \deff \,\
\frac{1}{2^{N-1}}\hspace{-5pt}
\sum_{u_{i+1}^N \in \{0,1\}^{N-i}} \hspace{-12pt}
\widetilde{W}\Bigl(y^N_1\hspace{1pt}{\bigm|}\hspace{1pt}
(u^{i-1}_1,u_i,u_{i+1}^N) \Bigr)
\end{split}
\ee
We also define the good bit-channels and bad bit-channels same as before i.e. for any $\beta \,{<}\, \shalf$ and $N = 2^n$
\begin{equation}
\begin{split}
\label{good-def2}
&\cG_N(W_1, W_2, \beta)\\
&{\deff}
\left\{\, i \in [N] ~:~ Z((W_1.W_2)^{(i)}_N) < 2^{-N^{\beta}}\!\!/N \hspace{1pt}\right\}
\\[0.250ex]
&\cB_N(W_1, W_2, \beta)\\
&{\deff}
\left\{\, i \in [N] ~:~ Z((W_1.W_2)^{(i)}_N) \ge 2^{-N^{\beta}}\!\!/N \hspace{1pt}\right\}
\\[-2.75ex]
\nonumber
\end{split}
\end{equation}

We show that the compound polar transformation shown in Figure\,\ref{scheme} is equivalent to two separated polar transformations of length $N/2$ for $W_1 \boxcoasterisk W_2$ and $W_1 \circledast W_2$ independently. Therefore, the channel polarization theorem can be established for this proposed transformation accordingly. This is proved next.
The following lemma is needed to establish the proof of \Tref{thm3}.
\begin{lemma}
\label{sum-capacity}
For any two BSM channels $W_1$ and $W_2$,
$$
\cC(W_1 \boxcoasterisk W_2) + \cC(W_1 \circledast W_2) = \cC(W_1) + \cC(W_2)
$$
\end{lemma}
\begin{proof}
Let $U_1$ and $U_2$ be two independent uniform binary random variables. Let $Y_1$ and $Y_2$ be the outputs of the channels $W_1$ and $W_2$ with inputs $U_1 + U_2$ and $U_2$ respectively, as depicted in Figure\,\ref{building_block}. Then
\begin{align*}
\cC(W_1) + \cC(W_2) &= I(U_1+U_2;Y_1) + I(U_2;Y_2) = I(U_1^{2};Y_{1}^{2})\\
 &= I(U_1;Y_1^{2}) + I(U_2;Y_1^2 | U_1) \\
 &= I(U_1;Y_1^2) + I(U_2;Y_1^2,U_1) \\
 &= \cC(W_1 \boxcoasterisk W_2) + \cC(W_1 \circledast W_2)
\end{align*}
where we used the fact that $W_1$, $W_2$, $W_1 \boxcoasterisk W_2$ and $W_1 \circledast W_2$ are all symmetric and therefore, the symmetric capacity is equal to the capacity for each of them. 
\end{proof}

\begin{theorem}
\label{thm3}
For any two BSM channels $W_1$ and $W_2$ and any constant $\beta \,{<}\, \shalf$ we have
$$
\lim_{N \to \infty} \frac{\left|\cG_N(W_1, W_2, \beta)\right|}{N}
\,=\, 
\frac{1}{2}\bigl(\cC(W_1) + \cC(W_2)\bigr)\vspace{1.5ex}
$$
\end{theorem}
\begin{proof}
For simplicity let $\dot{W}$ denote $W_1 \boxcoasterisk W_2$ and $\ddot{W}$ denote  $W_1 \circledast W_2$. By induction on $n = \log N$, we can show that for $1 \leq i \leq N/2$
\be{thm3-1}
(W_1.W_2)^{(i)}_N = \dot{W}^{(i)}_{N/2}
\ee
and for $N/2 < i \leq N$
\be{thm3-2}
(W_1.W_2)^{(i)}_N = \ddot{W}^{(i-N/2)}_{N/2}
\ee
The base of induction is clear by definition. The induction step follows by the recursive structure of the polar transformation. 

If $N$ is large enough, then we can pick $\beta'$ such that $\beta < \beta' < \shalf$ and
$$
\frac{2^{-(N/2)^{\beta'}}}{N/2} < \frac{2^{-N^{\beta}}}{N}
$$
Then definitions of good bit-channels given in \eq{good-def} and \eq{good-def2} together with \eq{thm3-1} and \eq{thm3-2} imply that
\be{thm3-3}
\left|\cG_N(W_1, W_2, \beta)\right| \geq \left|\cG_{N/2}(\dot{W},\beta')\right| +  \left|\cG_{N/2}(\ddot{W},\beta')\right|
\ee
Then by \Tref{thm1} and \Lref{sum-capacity},
\begin{align*}
&\lim_{N \to \infty} \frac{\left|\cG_N(W_1, W_2, \beta)\right|}{N} \\
& \geq \lim_{N \to \infty} \frac{\left|\cG_{N/2}(\dot{W},\beta')\right|}{N} +
\lim_{N \to \infty} \frac{\left|\cG_{N/2}(\ddot{W},\beta')\right|}{N}\\
&=  \frac{1}{2}\bigl(\cC(\dot{W}) + \cC(\ddot{W})\bigr) =
  \frac{1}{2}\bigl(\cC(W_1) + \cC(W_2)\bigr)
\end{align*}
On the other hand we have
\begin{align}
\frac{N}{2} \bigl(I(W_1) + I(W_2)\bigr) &= \sum^N_{i = 1} I\bigl( (W_1.W_2)^{(i)}_N \bigr) \label{thm3-4}\\
& \geq \sum_{i \in \cG_N(W_1, W_2, \beta)} I\bigl( (W_1.W_2)^{(i)}_N \bigr) \notag \\
& \geq  \sum_{i \in \cG_N(W_1, W_2, \beta)} 1 - Z\bigl( (W_1.W_2)^{(i)}_N \bigr) \label{thm3-5} \\
& \geq \left|\cG_N(W_1, W_2, \beta)\right| - 2^{-N^{\beta}} \label{thm3-6}
\end{align}
\eq{thm3-4} is by the chain rule on the mutual information between the input and output of the scheme shown in Figure\,\ref{scheme}. \eq{thm3-5} follows by \eq{Z-I}. \eq{thm3-6} holds by definition of the set of good bit-channels $\cG_N(W_1, W_2, \beta)$. Therefore,
$$
\lim_{N \to \infty} \frac{\left|\cG_N(W_1, W_2, \beta)\right|}{N} \leq \frac{1}{2}\bigl(\cC(W_1) + \cC(W_2)\bigr)
$$
which completes the proof of theorem.
\end{proof}

The encoding of our scheme is similar to that of polar codes. Let $k = \left|\cG_N(W_1, W_2, \beta)\right|$. Then the polar code associated with the set of good bit-channels $\cG_N(W_1, W_2, \beta)$ is a $(k,N)$ code. The positions corresponding to the indices in $\cG_N(W_1, W_2, \beta)$ carry the information bits and the rest of input bits are frozen to zeros. Then the following theorem and the corollary follows similar to \Tref{thm2} and \Cref{cor1} proved in \cite{Arikan}.
\begin{theorem}
\label{thm_error}
Suppose that a message $\bU$ is chosen uniformly at random from $\left\{0,1\right\}^k$, encoded using polar code associated with a set $\cA \subseteq [N]$ and transmitted over $W_1$ and $W_2$ as described in Figure\,\ref{scheme}. Then the probability that the received word $\bY$ is not decoded to $\bU$ under successive cancellation decoding satisfies
$$
\text{Pr}\left\{\hat{\bU} \neq \bU \right\} \leq \sum_{i \in \cA} Z\bigl((W_1.W_2)^{(i)}_N \bigr)
$$
\end{theorem}
\begin{corollary}
\label{cor2}
For any $\beta \,{<}\, \shalf$ and any two BSM channels $W_1$ and $W_2$, the polar code of length $N$ associated with the set of good bit-channels $\cG_N(W_1, W_2, \beta)$ defined in \eq{good-def2} approaches the average of the capacities of $W_1$ and $W_2$. Furthermore, the probability of frame error under successive cancellation decoding is less than $2^{-N^{\beta}}$. 
\end{corollary}

\subsection{Extending the construction to $l$-multi-channels}

In this section, we generalize the compound structure proposed in the foregoing subsection to the case of $l$-multi-channels. Suppose that a multi-channel consisting of set of $l$ B-DMCs $W_1,W_2,\dots,W_l$ is given. We fix an $l \times l$ invertible matrix $G_0$ as the initial matrix. Then the building block corresponding to $G_0$ is shown in Figure\,\ref{general_building}. In this figure, $x^l_1 = u^l_1.G_0$ and then $x_1,x_2,\dots,x_l$ are transmitted through $W_1,W_2,\dots,W_l$, respectively. 

\begin{figure}[h]
\begin{center}
\includegraphics[width=2in]{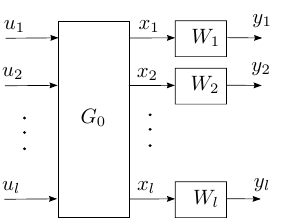}
\caption{The building block for the general case of $l$-multi-channels}
\label{general_building}
\end{center}
\end{figure}

In \cite{KSU}, a general transformation $\Gn$ is considered, where $G$ is an arbitrary $l \times l$ matrix with $l \geq 3$. A necessary and sufficient condition on $G$ is provided which guarantees polarization for any BSM channel. It is proved in \cite{KSU} that if $G$ is an invertible matrix, then polarization happens if and only if $G$ is not upper triangular. As a result, we can pick the matrix $G_0$ to satisfy this condition. Then the results of the forgoing subsection can be easily generalized to compound schemes of length $l^n$ transmitted over an $l$-multi-channel with transform matrix $G_0^{\otimes n}$ . However, the problem is that the successive cancellation decoder for this scheme is not easy to implement and decoding complexity grows by a factor of $2^l$. 

In the case that $l=2^m$ is a power of two, we pick the initial matrix for the building block to be $G_0 = G^{\otimes m}$, where G is the base $2 \times 2$ polarization matrix. Then the polarization matrix $G^{\otimes n}$ is applied to this building block resulting in a compound polar code of length $N = 2^{n+m}$. The advantage of picking this particular $G_0$ is the low complexity decoding algorithm. In fact, the successive cancellation decoder with complexity $O(N \log N)$ that is used for the Arikan's polar code of length $N$ can be applied to this compound code as well.

\subsection{$l$-compound polar codes with low complexity decoder}

In this section, for an arbitrary number of constituent channels $l$, we propose a scheme which enjoys the low complex $O(N \log N)$ decoder. 

For $n \geq 0$, we construct the general scheme with length $N=l.2^n$ as follows. We apply the polar transformation $G^{\otimes n}$ to the proposed building block in Figure\,\ref{general_building}. The block diagram of the proposed transformation is shown in Figure\,\ref{general_scheme}. In fact, the input sequence $u_1^N$ is multiplied by $G_0 \otimes G^{\otimes n}$. We design the permutation $\pi$ in such a way that the first $N/l$ encoded bits are transmitted through $N/l$ independent copies of $W_1$, the second $N/l$ encoded bits are transmitted through $N/l$ independent copies of $W_2$ etc. This is the general compound polar transformation.

\begin{figure}[h]
\begin{center}
\includegraphics[width=\linewidth]{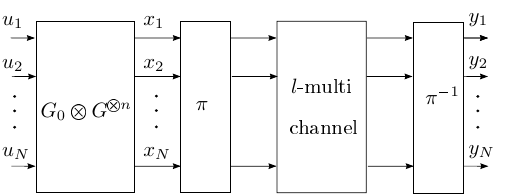}
\caption{The compound polar transformation for the general case of $l$-multi-channels}
\label{general_scheme}
\end{center}
\end{figure}

Suppose that $W_1,W_2,\dots,W_l$ are BSM channels. Then it can be shown that the compound polar transformation in Figure\,\ref{general_scheme} is equivalent to $l$ separated polar transformations of length $2^n$ for $l$ certain bit-channels corresponding to the building block, independently. Then the results of \Tref{thm_error} and \Cref{cor2} can be generalized to this compound scheme accordingly. 

The original successive cancellation (SC) decoder of polar codes invented by Arikan in \cite{Arikan} can be extended to the compound polar codes with some small modifications in a straightforward way. Let $N=l2^n$ and suppose that $u^N_1$ is the vector that is multiplied by $G_0 \otimes \Gn$ and then transmitted over independent copies of $W_1,W_2,\dots,W_l$ as shown in Figure\,\ref{general_scheme}. Let $y^N_1$ denote the received word. For $i=1,2,\dots, N$ if $W^{(i)}_N$ is not a good bit-channel, then the decoder knows that $i$-th bit $u_i$ is set to zero and  therefore, $\hat{u_i} = u_i = 0$. Otherwise, the decoder computes the likelihood $L^{(i)}_N$ of $u_i$, given the channel outputs $y^N_1$ and previously decoded $\hat{u}^{i-1}_1$. Then it makes the hard decision based on $L_N^{(i)}$.

The likelihood functions $L^{(i)}_N$ can be computed recursively similar to what Arikan proposed \cite{Arikan}. The only question is how to initiate the SC decoder for $n = 0$, when $N = l$. A naive way of computing transition probabilities of constituent bit-channels results in the complexity $O(2^l)$. The recursive steps can be done using Arikan's refined SC decoding algorithm with complexity $O(N\log N/l)$. Therefore, the total complexity is $O\bigl(N(\log N-\log l + 2^l)\bigr)$. As $N$ grows large, the dominating term is $N\log N$ and therefore the total complexity of SC decoding algorithm is $O(N\log N)$.

%=======================================================================%
%                                                                       %
%       4. Numerical Analysis of the proposed scheme                    %              
%                                                                       %
%=======================================================================%
\section{Simulation results}
\label{sec:four}

Transmission over AWGN channel with $16$-QAM BICM is considered. In $16$-QAM BICM, there are actually two constituent binary-input channels. Among the $4$ bits in each symbol, two of them go through one channel denoted by $W_1$ and the other two goes through the other channel denoted by $W_2$. In the constellation that we are using, the mapping is such that bits 1 and 2 in each symbol goes through $W_1$ and bits 3 and 4 goes through $W_2$. The conventional way of modulating a codeword of length $N$ is to split into $N/4$ sub-blocks of 4 consecutive bits each and map them into $N/4$ symbols. 

For the code construction, we take a numerical approach to estimate to probability of error of the individual bit-channels. For simulation, the block length is fixed to $2^{10} = 1024$ and the rate to $1/2$. Transmission over AWGN channel is considered. In the separated scheme, we split the bits into two groups based on the channel that they observe. Then we construct two polar codes for $W_1$ and $W_2$ separately. This should be done in such a way that the total rate is $1/2$. Since the two channels are different, we have to figure out how to assign the rate to be transmitted on each channel. We fix the $E_b/N_0 = 5$dB and then numerically estimate the bit-channel probability of error for each of the channels. We pick the rates such that the total probability of error is minimized. It turns out that the rate $0.62$ on the stronger channel $W_1$ and $0.38$ on the weaker channel $W_2$ minimizes the total probability of error at $5$ dB. We use the same scheme for all SNRs. 

For the compound scheme, we use an interleaver to guarantee the right ordering of the transmitted bits as depicted in Figure\,\ref{scheme}. In fact, the interleaver switches the second and third bit in each symbol with each other. There is a deinterleaver at the decoder which does the same to the channel outputs. The comparison between the two methods is shown in Figure\,\ref{plot3}. As we can observe, our compound scheme is about $1.5$dB better than the separated scheme at moderate SNR's. As SNR increases, the gaps become larger and the curves start diverging. 

\begin{figure}[h]
\centering
% This file was created by matlab2tikz v0.0.5.
% Copyright (c) 2008--2010, Nico Schlömer <nico.schloemer@ua.ac.be>
% All rights reserved.
%
% The latest updates can be retrieved from
%  http://win.ua.ac.be/~nschloe/content/matlab2tikz/
% and
%  http://www.mathworks.com/matlabcentral/fileexchange/22022 .
% where you can also make suggestions and rate matlab2tikz.

\begin{tikzpicture}

% Axis at [0.13 0.11 0.78 0.81]
\begin{semilogyaxis}[
scale only axis,
width = 2.5in,
height = 1.875in,
xmin=2, xmax=8,
ymin=1e-05, ymax=1,
xlabel={$E_b/N_0$ [dB]},
ylabel={Block error rate},
title = {$N = 1024$, rate $= 1/2$, $16$-QAM},
xmajorgrids,
ymajorgrids,
yminorgrids,
legend entries={{\small separated polar code}, {\small compound polar code} },
legend style={nodes=right}]

\addplot [
color=red,
solid,
mark=x,
mark options={solid}
]
coordinates{
 (3,0.9174)
 (3.5,0.7519)
 (4,0.4348)
 (4.5,0.188)
 (5,0.0611)
 (5.5,0.0184)
 (6,0.0048)
 (6.5,0.0015)
 (7,0.00043946)
 (7.5,9.7561e-05)

};

\addplot [
color=blue,
solid,
mark=x,
mark options={solid}
]
coordinates{
 (2,0.99)
 (2.5,0.97)
 (3,0.641)
 (3.5,0.3509)
 (4,0.1107)
 (4.5,0.0247)
 (5,0.0032)
 (5.5,0.00036449)
 (6,4.2146e-05)

};

\end{semilogyaxis}

\end{tikzpicture}
\caption{Performance of the proposed scheme over 16-QAM}
\label{plot3}
\end{figure}
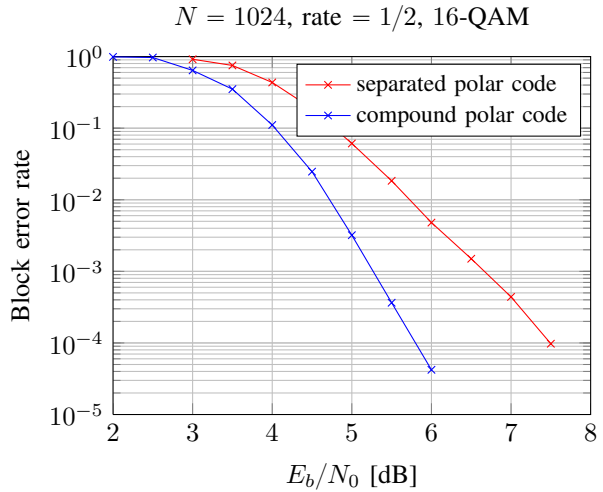

%=======================================================================%
%                                                                       %
%       5. Conclusion and Future Work                    %              
%                                                                       %
%=======================================================================%
\section{Conclusions and Future Work}
\label{sec:five}

In this work, we proposed a compound polar-based scheme to be used over multi-channels. We extended the channel polarization theorem to this case and proposed the compound polar code construction. We also provided simulation results at finite block lengths for BICM channels. There are a couple of open problems regarding the proposed construction. One is regarding the code construction over BICM channels. We took a numerical simulation-based approach to estimate the bit-channel probability of errors. An alternative way is to extend the Tal-Vardy method \cite{TV} for efficiently constructing polar codes to the case of BICM channels. This is left as a future work. Another question is what is the best ordering of the channels in the general building block in Figure\,\ref{general_building}. It turns out that for the case of $2$-multi-channels the ordering does not matter. However, for general $l$ different orderings may result in different polarization rates. The question is how to characterize the polarization rate in terms of the building block and how to pick the best ordering. The second open problem is about choosing the initial $l \times l$ matrix $G_0$. Which $G_0$ results in the best performance at finite block lengths? One strategy is to pick $G_0$ that maximizes the \emph{polarization rate} among all $l \times l$ channels as characterized in \cite{KSU}. However, since $G_0$ is only used in one level of polarization in the compound polar code, this is not necessarily the best choice. The answer to these questions will help to design more efficient schemes at finite block lengths.

\bibliographystyle{IEEEtran}

\end{document}